\documentclass[11pt]{article}

\usepackage{authblk}
\usepackage{amsthm}
\usepackage{mathtools}
\usepackage{comment}
\usepackage{amsfonts}
\usepackage{amsmath}
\usepackage[margin=1in]{geometry}
\usepackage{graphicx} % Required for inserting images
\usepackage{amsmath,amsthm,amssymb}
\usepackage{thmtools, thm-restate}
\usepackage[ruled]{algorithm} % SN: I think this is what lipics recommends
\usepackage{algpseudocode}
\usepackage{hyperref}   % provides \href and \url
\usepackage{url}        % provides \path (or \usepackage{hyperref} alone may suffice)

\declaretheorem{theorem}
\newtheorem{lemma}{Lemma} 
 
\newtheorem{corollary}{Corollary}
\newtheorem{definition}{Definition}

\newtheorem{remark}{Remark} 

\title{Resource bounded Ku\v{c}era--G\'{a}cs Theorems} 

\author[1]{Satyadev Nandakumar}
\author[1]{Akhil S}
\author[1]{Chandra Shekhar Tiwari}
\affil[1]{
	Department of Computer Science and Engineering\\
	Indian Institute of Technology Kanpur,
	Kanpur, Uttar Pradesh, India.
}
\affil[1]{\{satyadev,akhis,chandrat\}@cse.iitk.ac.in}

\begin{document}

% Decision Probtheorems 
%\newcommand{\SAT}{\mathrm{SAT}}
%\newcommand{\MAXSAT}{\mathrm{MAXSAT}}
%\newcommand{\MAXTSAT}{\mathrm{MAX3SAT}}
\newcommand{\MAXSAT}{\textsc{MaxSat}}
\newcommand{\MAXTSAT}{\textsc{Max3Sat}}
\newcommand{\SAT}{{\rm SAT}}
\newcommand{\MAXCLIQUE}{\textsc{MaxClique}}
\newcommand{\CLIQUE}{\textsc{Clique}}

% Proof Commands
\newcommand{\PROOF}{\begin{proof}}%{{\noindent\bf{\em Proof:\/}~~}}
\newcommand{\PROOFSKETCH}{{\noindent\bf{\em Proof sketch:\/}~~}}
\newcommand{\QED}{\end{proof}}
\newcommand{\QEDbox}{\nolinebreak{\hfill $\square$}\vspace*{0.3cm}}

% Miscellaneous
\newcommand{\third}{3^{\mathrm{rd}}}
\newcommand{\second}{2^{\mathrm{nd}}}
\newcommand{\first}{1^{\mathrm{st}}}
\newcommand{\Prob}{\mathrm{Prob}}
\newcommand{\floor}[1]{ \left\lfloor #1 \right\rfloor }
\newcommand{\ceil}[1]{ \left\lceil #1 \right\rceil }
\newcommand{\range}{\mathrm{range}}
\newcommand{\myset}[2]{ \left\{ #1 \left\;\;:\;\; #2 \right. \right\} }
\newcommand{\mysetl}[2]{ \left\{\left. #1 \right| #2 \right\} }
\newcommand{\bval}[1]{[\![ #1 ]\!]}
\newcommand{\prefix}{\sqsubseteq}
\newcommand{\prprefix}{\underset{\neq}{\sqsubset}}
\newcommand{\diam}{\mathrm{diam}}
\newcommand{\binary}{\{0,1\}^*}
\newcommand{\limn}{\lim\limits_{n\to\infty}}
\newcommand{\liminfn}{\liminf\limits_{n\to\infty}}
\newcommand{\limsupn}{\limsup\limits_{n\to\infty}}
\newcommand{\ith}{{i^{\mathrm{th}}}}
\renewcommand{\Pr}[2]{\underset{#1}{\mathrm{Prob}}\left[#2\right]}
\newcommand{\size}[1]{\parallel\negthinspace #1\negthinspace\parallel}
\newcommand{\restr}{\upharpoonright}
\newcommand{\vb}{{\vec{\beta}}}
\newcommand{\RANDb}{\RAND^\vb}

\newcommand{\join}{\oplus}
\newcommand{\TQBF}{{\rm TQBF}}
\newcommand{\GI}{{\rm GI}}
\newcommand{\parityP}{{\oplus\P}}
\newcommand{\sharpP}{{\#\P}}
\newcommand{\CeqP}{{{\rm C}_=\P}}
\newcommand{\MODkP}{{{\rm MOD}_k\P}}
\newcommand{\BP}{\mathrm{BP}\cdot}

% Blackboard Bold Sets
\newcommand{\N}{\mathbb{N}}
\newcommand{\Z}{\mathbb{Z}}
\newcommand{\Q}{\mathbb{Q}}
\newcommand{\R}{\mathbb{R}}

% Measure and Dimension
\newcommand{\p}{{\mathrm{p}}}
\newcommand{\ptwo}{{\p_{\thinspace\negthinspace_2}}}
\newcommand{\pthree}{{\p_{\thinspace\negthinspace_3}}}
\newcommand{\pspace}{{\mathrm{pspace}}}
\newcommand{\ptwospace}{{\mathrm{p}_{\thinspace\negthinspace_2}\mathrm{space}}}
\newcommand{\poly}{{\mathrm{poly}}}
\newcommand{\rec}{{\mathrm{rec}}}
\newcommand{\all}{{\mathrm{all}}}
\newcommand{\FS}{{\mathrm{FS}}}
\newcommand{\comp}{{\mathrm{comp}}}
\newcommand{\constr}{{\mathrm{constr}}}
\newcommand{\mup}{\mu_{\mathrm{p}}}
\newcommand{\dimall}{\dim_\all}
\newcommand{\dimp}{\dim_\p}
\newcommand{\dimfs}{\mathrm{dim}_\mathrm{FS}}
\newcommand{\dimpspace}{\mathrm{dim}_\mathrm{pspace}}
\newcommand{\dimh}{\mathrm{dim}_\mathrm{H}}
\newcommand{\dimrec}{\mathrm{dim}_\mathrm{rec}}
\newcommand{\dimcomp}{\mathrm{dim}_\comp}

%% NOTE: Change to cdim if required.
\newcommand{\cdim}{\mathrm{dim}}

\newcommand{\dimptwo}{\dim_\ptwo}
\newcommand{\dimpthree}{\dim_\pthree}
\newcommand{\dimptwospace}{\dim_\ptwospace}
\newcommand{\Dimptwo}{\Dim_\ptwo}
\newcommand{\Dimptwospace}{\Dim_\ptwospace}
\newcommand{\muptwo}{\mu_\ptwo}
\newcommand{\mupthree}{\mu_\pthree}
\newcommand{\Deltapthree}{{\Delta^\p_3}}
\newcommand{\dimDeltapthree}{\dim_\Deltapthree}
\newcommand{\Deltaptwo}{{\Delta^\p_2}}
\newcommand{\dimDeltaptwo}{\dim_\Deltaptwo}
\newcommand{\DeltaPtwo}{{\Delta^\P_2}}
\newcommand{\DeltaPthree}{{\Delta^\P_3}}

\renewcommand{\dim}{{\mathrm{dim}}}
\newcommand{\udim}{{\overline{\dim}}}
\newcommand{\ldim}{{\underline{\dim}}}
\newcommand{\dimpack}{{\dim}_{\mathrm{P}}}
\newcommand{\dimb}{{\dim}_{\mathrm{B}}}
\newcommand{\udimb}{{\udim}_{\mathrm{B}}}
\newcommand{\ldimb}{{\ldim}_{\mathrm{B}}}
\newcommand{\udimmb}{\udim_{\mathrm{MB}}}

\newcommand{\Dim}{{\mathrm{Dim}}}
\newcommand{\cDim}{{\mathrm{cDim}}}
\newcommand{\Dimfs}{{\Dim_\FS}}
\newcommand{\Dimp}{{\Dim_\p}}
\newcommand{\Dimcomp}{{\Dim_\comp}}
\newcommand{\Dimpspace}{{\Dim_\pspace}}

\newcommand{\str}{{\mathrm{str}}}
\newcommand{\regSS}{S^\infty}
\newcommand{\strSS}{S^\infty_{\mathrm{str}}}

\newcommand{\Ghat}{{\widehat{\G}}}
\newcommand{\hatG}{\Ghat}
\newcommand{\Gstr}{\G^{\mathrm{str}}}
\newcommand{\hatGstr}{\hatG^{\mathrm{str}}}

\newcommand{\D}{{\mathcal{D}}}
\newcommand{\Gconstr}{\mathcal{G}_\mathrm{constr}}
\newcommand{\hatGconstr}{\mathcal{\hatG}_\mathrm{constr}}
\newcommand{\Gstrconstr}{\mathcal{G}_\mathrm{constr}^\mathrm{str}}
\newcommand{\hatGstrconstr}{\mathcal{\hatG}_\mathrm{constr}^\mathrm{str}}

\newcommand{\bd}{{\mathbf{d}}}
\newcommand{\FREQ}{{\mathrm{FREQ}}}
\newcommand{\freq}{{\mathrm{freq}}}

% Prediction
\newcommand{\pred}{\mathrm{pred}}
\newcommand{\Pred}{\mathrm{Pred}}
\newcommand{\Predp}{\mathrm{Pred}_\p}
\newcommand{\predp}{\mathrm{pred}_\p}
\newcommand{\unpred}{\mathrm{unpred}}
\newcommand{\Unpred}{\mathrm{Unpred}}
\newcommand{\unpredlog}{\mathrm{unpred}^\mathrm{log}}
\newcommand{\Unpredlog}{\mathrm{Unpred}^\mathrm{log}}
\newcommand{\unpredlogp}{\mathrm{unpred}^\mathrm{log}_\mathrm{p}}
\newcommand{\unpredabs}{\mathrm{unpred}^\mathrm{abs}}
\newcommand{\unpredabsp}{\mathrm{unpred}^\mathrm{abs}_\mathrm{p}}
\newcommand{\predcomp}{\pred_\comp}
\newcommand{\predfs}{\pred_\mathrm{FS}}
\newcommand{\Ik}{{\cal I}_k}
\newcommand{\Ikbar}{\overline{\Ik}}
\newcommand{\Hk}{{\cal H}_k}

\newcommand{\calL}{{\mathcal{L}}}
\newcommand{\logL}{\calL^{\log}}
\newcommand{\strlogL}{\calL^{\log}_{\str}}
\newcommand{\abs}{\mathrm{abs}}
\newcommand{\absL}{\calL^{\abs}}

\newcommand{\loss}{{\mathit{loss}}}
\newcommand{\losslog}{\loss^{\log}}
\newcommand{\lossabs}{\loss^\abs}

% Complexity Classes
\newcommand{\co}[1]{\mathrm{co}#1}
\newcommand{\strings}{\{0,1\}^*}
\newcommand{\ALL}{\mathrm{ALL}}
\newcommand{\RAND}{\mathrm{RAND}}
\newcommand{\RANDweak}{\mathrm{RAND}_\mathrm{weak}}
\newcommand{\DIM}{\mathrm{DIM}}
\newcommand{\DIMstr}{\DIM_\str}
\newcommand{\coDIM}{\mathrm{coDIM}}
\newcommand{\REC}{\mathrm{REC}}
\newcommand{\DEC}{\mathrm{DEC}}
\newcommand{\CE}{\mathrm{CE}}
\newcommand{\coCE}{\mathrm{coCE}}
\newcommand{\RE}{\mathrm{RE}}
\newcommand{\coRE}{\mathrm{coRE}}
\newcommand{\SPARSE}{\mathrm{SPARSE}}
\newcommand{\DENSE}{\mathrm{DENSE}}
\newcommand{\DENSEio}{\mathrm{DENSE}_\mathrm{i.o.}}
\newcommand{\DTIME}{\mathrm{DTIME}}
\newcommand{\NTIME}{\mathrm{NTIME}}
\newcommand{\DTIMEF}{\mathrm{DTIMEF}}
\newcommand{\NP}{{\ensuremath{\mathrm{NP}}}} 
\newcommand{\coNP}{\co{\NP}}
\newcommand{\AM}{{\rm AM}}
\newcommand{\MA}{{\rm MA}}
\newcommand{\DeltaPTwo}{\Delta^{\mathrm{P}}_2}
\newcommand{\DP}{\mathrm{DP}}
\newcommand{\SUBEXP}{\mathrm{SUBEXP}}
\newcommand{\SIZE}{{\rm SIZE}}
\newcommand{\SIZEio}{\SIZE_{\rm i.o.}}
\newcommand{\NSIZE}{{\rm NSIZE}}
\newcommand{\NSIZEio}{\NSIZE_{\rm i.o.}}
\newcommand{\BPP}{{\rm BPP}}
\newcommand{\PF}{{\rm PF}}
\newcommand{\ZPP}{{\rm ZPP}}
\newcommand{\RP}{{\rm RP}}
\newcommand{\coRP}{\co{\RP}}
\newcommand{\APP}{{\rm APP}}
\newcommand{\AP}{{\rm AP}}
\newcommand{\E}{{\rm E}}
\newcommand{\NE}{{\rm NE}}
\newcommand{\EE}{{\rm EE}}
\newcommand{\ESPACE}{{\rm ESPACE}}
\newcommand{\PSPACE}{{\rm PSPACE}}
\renewcommand{\P}{\ensuremath{{\mathrm P}}}
\newcommand{\QP}{{\rm QP}}
\newcommand{\NQP}{{\rm NQP}}
\newcommand{\Etwo}{{\rm E_2}}
\newcommand{\Ethree}{{\rm E_3}}
\newcommand{\EtwoSPACE}{{\rm E_2SPACE}}
\newcommand{\EXP}{{\rm EXP}}
\newcommand{\NEXP}{{\rm NEXP}}
\newcommand{\DSPACE}{{\rm DSPACE}}
\newcommand{\DSPACEF}{{\rm DSPACEF}}
\newcommand{\NEE}{{\rm NEE}}
\newcommand{\Poly}{{\rm P}}
\newcommand{\Pc}{{\rm P}}
\newcommand{\PPoly}{{\rm {P/Poly}}}
\newcommand{\Psel}{{\rm {P-sel}}}
\newcommand{\WS}{{\rm WS}}
\newcommand{\Ei}{\E_i}
\newcommand{\EiSPACE}{\mathrm{E}_i\mathrm{SPACE}}
\newcommand{\EoneSPACE}{\mathrm{E}_1\mathrm{SPACE}}
\newcommand{\linear}{\mathrm{linear}}
\newcommand{\polynomial}{\mathrm{polynomial}}
\newcommand{\EXPSPACE}{\mathrm{EXPSPACE}}
\newcommand{\PP}{{\rm PP}}
\newcommand{\SPP}{{\rm SPP}}
\newcommand{\GapP}{{\rm GapP}}
\newcommand{\PH}{{\rm PH}}
\newcommand{\Ppoly}{{\rm P/poly}}
\newcommand{\EH}{{\rm EH}}
\newcommand{\EHsigmaone}{{\Sigma_1^\E}}
\newcommand{\EHsigmatwo}{{\Sigma_2^\E}}
\newcommand{\EHsigmathree}{{\Sigma_3^\E}}
\newcommand{\EHdeltaone}{{\Delta_1^\E}}
\newcommand{\EHdeltatwo}{{\Delta_2^\E}}
\newcommand{\EHdeltathree}{{\Delta_3^\E}}
\newcommand{\EHpione}{{\Pi_1^\E}}
\newcommand{\EHpitwo}{{\Pi_2^\E}}
\newcommand{\EHpithree}{{\Pi_3^\E}}
% Calligraphy
\newcommand{\calA}{{\cal A}}
\newcommand{\calB}{{\cal B}}
\newcommand{\calC}{{\cal C}}
\newcommand{\calD}{{\cal D}}
\newcommand{\F}{{\mathcal{F}}}
\newcommand{\X}{{\mathcal{X}}}
\newcommand{\V}{{\mathcal{V}}}
\newcommand{\calH}{{\cal H}}
\newcommand{\calP}{{\cal P}}

% Kolmogorov Complexity
\newcommand{\calKS}{{\cal KS}}
\newcommand{\calKT}{{\cal KT}}
\newcommand{\calK}{{\cal K}}
\newcommand{\calCD}{{\cal CD}}
\newcommand{\calCND}{{\cal CND}}
\newcommand{\K}{{\mathrm{K}}}
\newcommand{\KS}{{\mathrm{KS}} }
\newcommand{\KT}{{\mathrm{KT}} }
\newcommand{\CND}{{\mathrm{CND}} }

% Reductions
\newcommand{\leqpm}{\leq^\p_\mathrm{m}}
\newcommand{\leqmp}{\leqpm}
\newcommand{\leqpT}{\leq^\mathrm{p}_\mathrm{T}}
\newcommand{\leqpr}{\leq^\mathrm{p}_r}
\newcommand{\leqonettp}{\leq_{1-\mathrm{tt}}^{\mathrm{P}}}
\newcommand{\murec}[1]{{\mu({#1}\mid \REC)}}
\newcommand{\muex}[1]{{\mu({#1}\mid \Etwo)}}
\newcommand{\redu}[2]{{\leq^{\rm {{#1}}}_{{\rm{#2}}}}}
\newcommand{\predu}[1]{{\redu{P}{#1}}}
\newcommand{\pmred}{{\predu{m}}}
\newcommand{\ptred}{{\predu{T}}}

%pair reductions

\newcommand{\leqppm}{\leq^{\p\p}_\mathrm{m}}
\newcommand{\leqppum}{\leq^{\p\p}_\mathrm{um}}

\newcommand{\leqppsm}{\leq^{\p\p}_\mathrm{sm}}
\newcommand{\leqppT}{\leq^{\mathrm{pp}}_\mathrm{T}}
\newcommand{\leqppuT}{\leq^{\mathrm{pp}}_\mathrm{uT}}

% Reduction Closures
\newcommand{\Pm}{\mathrm{P}_{\mathrm{m}}}
\newcommand{\PT}{\mathrm{P}_{\mathrm{T}}}
\newcommand{\Ptt}{\mathrm{P}_{\mathrm{tt}}}
\newcommand{\Pctt}{\mathrm{P}_{\mathrm{ctt}}}
\newcommand{\Pdtt}{\mathrm{P}_{\mathrm{dtt}}}
\newcommand{\Ponett}{\mathrm{P}_{1-\mathrm{tt}}}
\newcommand{\Pmi}{\Pm^{-1}}
\renewcommand{\Pr}{\mathrm{P}_r}
\newcommand{\Pri}{\Pr^{-1}}
% Degrees
\newcommand{\degpm}{\mathrm{deg}^{\mathrm{p}}_{\mathrm{m}}}
\newcommand{\degpr}{\mathrm{deg}^\mathrm{p}_r}
\newcommand{\equivpm}{\equiv^\p_\mathrm{m}}
% Hard and Complete sets
\newcommand{\m}{\mathrm{m}}
\newcommand{\Cpm}{\calC^\p_\m}
\newcommand{\Hpm}{\calH^\p_\m}
\newcommand{\Cpr}{\calC^\p_r}
\newcommand{\Hpr}{\calH^\p_r}

% Stuff from CPED
\newcommand{\CH}{\mathcal{H}}
\newcommand{\CHdec}{\mathcal{H}_\DEC}
\newcommand{\CHce}{\mathcal{H}_\CE}
\newcommand{\Pione}{\Pi^0_1}
\newcommand{\Pitwo}{\Pi^0_2}
\newcommand{\BorelPi}[2]{{\bf \Pi}^ #1 _#2}
\newcommand{\BorelSigma}[2]{{\bf \Sigma}^ #1_#2}
\newcommand{\Sigmaone}{\Sigma^0_1}
\newcommand{\Sigmatwo}{\Sigma^0_2}
\newcommand{\PioneC}{\Pione}% \upharpoonright \C}
\newcommand{\PitwoC}{\Pitwo}% \upharpoonright \C}
\newcommand{\SigmaoneC}{\Sigmaone}% \upharpoonright \C}
\newcommand{\SigmatwoC}{\Sigmatwo}% \upharpoonright \C}
\newcommand{\io}{{\mathrm{i.o.}}}
\newcommand{\aev}{{\mathrm{a.e.}}}
\newcommand{\ioA}{A^\io}
\newcommand{\ioB}{B^\io}
\newcommand{\ioC}{C^\io}
\newcommand{\aeA}{A^\aev}
\newcommand{\aeB}{B^\aev}

% Names
\newcommand{\MartinLof}{Martin-L\"{o}f}

\newcommand{\nth}{n^{\mathrm{th}}}

\newcommand{\supp}{\mathrm{supp}}
\newcommand{\eH}{\mathrm{H}}

% Dimension Entropy Rate
\newcommand{\SpanP}{\mathrm{SpanP}}

% Comments
% \usepackage[svgnames]{xcolor}
% \newcommand{\comment}[1]{{\color{DarkGreen} #1}}

% Contined Fractions
\newcommand{\denominator}{\text{D}}
\newcommand{\numerator}{\text{N}}

%These commands are provided for left-closed right-open intervals and
%left-open right-closed intervals, since writing intervals with
%different brackets fouls up Emacs filling
\newcommand{\lcro}[2]{[#1,#2)}
\newcommand{\lorc}[2]{(#1,#2]}
\renewcommand{\L}{\mathcal{L}}

\newcommand{\Rplus}{[0,\infty)}
\renewcommand{\F}{\mathcal{F}}
\newcommand{\U}{\mathcal{U}}

\newcommand{\Kpoly}{\mathcal{K}_{\mathsf{poly}}}
\newcommand{\cdimp}{{\cdim}_\mathrm{P}}

\newtheorem{construction}{Construction}[section]
\newcommand{\exact}{}

\maketitle

\begin{abstract}
	 The Ku\v{c}era--G\'{a}cs  theorem \cite{Kucera84} \cite{Gacs86} is a fundamental result in algorithmic randomness. It states that every infinite
	sequence $X$ is Turing reducible to a Martin-L\"of random $R$.
	This paper studies resource-bounded analogues of the Kučera-Gács Theorem, at the resource bounds of polynomial-time and finite-state computation.
	
	We prove a \emph{quasi-polynomial-time}{ Kučera-Gács Theorem}, showing that every infinite sequence $X$ is quasi-polynomial-time reducible to a \emph{polynomial-time random} sequence $R$. We also show that for any $X$, the oracle use of $R$ is $n+o(n)$ bits for obtaining the first $n$ bits of $X$.

	We then study the relationship between compressibility and Turing reductions, in the polynomial-time setting.
	We establish that $\rho^-_{\mathsf{poly}}(X) = \mathcal{K}_{\mathsf{poly}}(X)$, demonstrating that the lower polynomial-time Turing decompression ratio is precisely characterized by the polynomial-time Kolmogorov complexity rate. We note that this characterization \emph{fails} for the polynomial-time dimension ($\cdimp$) if one-way functions exist, resolving an open problem from Doty's \cite{Doty08} work.

	We use these results to strengthen the {quasi-polynomial-time}{ Kučera-Gács Theorem}. We show that every infinite sequence $X$ is quasi-polynomial-time reducible to a {polynomial-time random} sequence $R$, where the lower oracle use rate of the reduction is less than $\Kpoly(X)$.

	We also establish a lower bound for how powerful the reduction needs
	to be if we have to extract any sequence from randoms. We show that any sequence  extracted from the (even larger) set of \emph{normal sequences} by a finite-state reduction must have a convergent asymptotic frequency for its symbols. Since sequences lacking this invariant property exist, they cannot be finite-state reduced from any normal sequence. Hence we show that the Kučera-Gács theorem \emph{fails} for finite-state reductions.
\end{abstract}
\newpage

\section{Introduction}
There is a fundamental dichotomy between structure and randomness.
Kolmogorov utilized computability theory to provide a new foundation
for information and probability theory by introducing the concept of an
``individual random string''. His key insight was to use
incompressibility as a basis for randomness, after fixing a universal
mode of description-algorithms. Kolmogorov defined a finite string
$x$ to be incompressible if the shortest program that outputs $x$ is
at least as long as $x$ itself. Martin-L\"of extended this approach to
define an individual random \emph{infinite} sequence. This concept is robust - it has equivalent
definitions using constructive measure theory, unpredictability,
compressibility, betting (martingales), and others (see, for example, Downey,
Hirschfeldt \cite{Downey10}, and Nies \cite{Nies2009}).

A striking phenomenon emerges when we examine how Martin-L\"of randoms
interact with Turing reducibility. Independently, Ku\v{c}era
\cite{Kucera84} and G\'{a}cs \cite{Gacs86} show that every infinite
sequence $X$ is Turing reducible to a Martin-L\"of random $R$. This
 result has led to a rich literature studying the
interplay between the strength of reductions and notions of
randomness.  G\'{a}cs' original result \cite{Gacs86} had a redundancy of $o(n)$, with  
$n+\sqrt{n} \log(n)$ oracle usage to determine the first $n$ bits of $X$ in the
reduction. Downey and Hirschfeldt \cite{Downey10} show that we cannot
improve this to $n+O(1)$.  Doty \cite{Doty08} established a deep
connection between the oracle usage and the effective Hausdorff
dimension - informally, the information rate - of $X$, showing that
the latter is precisely characterized by the oracle usage in the
Ku\v{c}era--G\'{a}cs theorem.

Resource-bounded reducibilities have also been studied. Book, Lutz and
Wagner \cite{Book1996} show that the characteristic sequence of a
decidable language $L$ is polynomial-time reducible to Martin-L\"of
randoms if and only if $L \in \BPP$. Thus, not every sequence is polynomial-time reducible to a Martin-L\"of random. This naturally raises the question: \emph{does an analogue of the
Kučera--Gács theorem hold when both randomness and reducibility are
restricted by feasible computation? }
Doty \cite{Doty08}
shows that his characterization of dimension also extends to the
oracle use in $\PSPACE$-bounded reducibility to $\PSPACE$ randoms, and
establishes a one-sided inequality for polynomial-time. However, the converse for
polynomial-time reduction remained open.
In this work, we resolve this gap. 

We first show that every sequence
is quasi-polynomial-time reducible to some polynomial-time random. Moreover, the reduction uses only $n+o(n)$ oracle bits to compute
the first $n$ bits of the sequence.
We then characterize one notion of
polynomial-time dimension, namely, that using time-bounded
Kolmogorov complexity rate using polynomial-time reductions to some (not necessarily random) sequence. In a recent work, Akhil, Nandakumar,
Pulari and Sarma \cite{OWFDimension} establish that notions of
polynomial-time dimension using betting algorithms called gales and
that using Kolmogorov complexity rates are \emph{inequivalent} in
polynomial-time, if cryptographic one-way functions exist. Thus Doty's
theorem \emph{does not} extend to polynomial-time dimension if one-way functions
exist, even though we obtain an alternate characterization using
Kolmogorov complexity rates.

We also establish a lower bound for how powerful the reduction needs
to be if we have to extract any sequence from randoms. We show that for 
finite-state reductions between infinite sequences,
even when we weaken the notion of randomness to include all normal
sequences (sequences where every finite string $w$ appears with
asymptotic frequency $2^{-|w|}$ ), we cannot extract every sequence. We
establish an \emph{invariant property} maintained by finite-state
reductions: if $X$ is finite-state reducible to a normal sequence $N$,
then $X$ must have a well-defined asymptotic frequency for all letters. But, it is easy to construct sequences where the
asymptotic frequency of a letter is undefined. Hence, such sequences
cannot be finite-state reducible to a normal sequence (and \emph{a fortiori}, to polynomial-time randoms).

\subsection{Our contributions.}
This work studies the extent to which the Ku\v{c}era--G\'{a}cs theorem persists under
computational resource bounds, with a particular focus on polynomial-time and
finite-state reducibilities.
We obtain the following results.
\begin{itemize}
	\item \textbf{Quasi-polynomial-time Ku\v{c}era--G\'{a}cs.}
	We show that for every infinite sequence $X$ there exists a polynomial-time random
	sequence $R$ such that $X$ is computable from $R$ via a quasi-polynomial-time reduction.
	Moreover, only $u_n = n+o(n)$ bits of the random oracle $R$ are required to compute the
	first $n$ bits of any given sequence $X$.

    \item \textbf{Polynomial-time decompression ratio characterisation.}
	We establish that the lower
	polynomial-time decompression ratio $\rho^-_{\mathsf{poly}}$ coincides with the polynomial-time Kolmogorov complexity rate $\mathcal{K}_{\mathsf{poly}}$. We show that $\rho^-_{\mathsf{poly}}$ and polynomial-time dimension $\cdimp$ are not equal in general, under the assumption that one-way functions exist. This resolves an open question from Doty \cite{Doty08}.
	
	\item \textbf{Dimension-sensitive Ku\v{c}era--G\'{a}cs.}
	We further strengthen the quasi-polynomial-time Ku\v{c}era--G\'{a}cs theorem by
	showing that the oracle rate use can be bounded by the $\mathcal{K}_{\mathsf{poly}}$
	dimension of the sequence, that is,
	$\liminf_{n\to\infty} u_n/n \le \mathcal{K}_{\mathsf{poly}}(X)$.
	
	\item \textbf{No Ku\v{c}era--G\'{a}cs for finite-state reductions.}
	We show that no analogue of the Ku\v{c}era--G\'{a}cs theorem holds for finite-state reductions. Even when we start from the set of normal numbers (the weakest notion of randomness), the sequences generated by a finite-state reduction must have well-defined limiting symbol frequencies, and therefore cannot generate all sequences.
	
\end{itemize}

\section{Preliminaries}

\subsection{Notation}
$\N$ denotes the set of natural numbers, $\Q$ denotes the set of rationals, and $\R$ denotes the set of real numbers.
$\Sigma$ denotes a non-empty finite alphabet, which
contains at least two symbols. $\Sigma^*$ denotes the set of all finite
strings over $\Sigma$. For a finite string $w \in \Sigma^*$,
$\lvert w \rvert$ denotes the length of $w$. For $n \in \mathbb{N}$, $\Sigma^n$ denotes the set of strings of length exactly $n$ over $\Sigma$. 
We write $\Sigma^{<n} $ for the set of strings of length strictly less than $n$, and 
$\Sigma^{\leq n}$ for the set of strings of length at most $n$.
$\Sigma^\infty$ denotes
the set of all infinite strings over $\Sigma$. $\lambda$ denotes the
empty string. For two strings $u, v \in \Sigma^*$, $u\cdot v$ represents
the string formed by concatenating $v$ after $u$.
For an infinite string $X \in \Sigma^\infty$, $X[i]$ denotes the
$i^{\text{th}}$ bit of $X$, with the indices starting from $0$. $X \restr n$ denotes the first $n$ bits
of $X$. For an infinite string $X=a_0\cdot a_1\cdot a_2\cdot a_3 \dots$ (where
$a_i \in \Sigma$) and $m \geq n$, $X[n:m]$ denotes the substring
$a_n\cdot a_{n+1}\cdot a_{n+2} \dots a_{m}$. We use the same notation for
finite strings $x \in \Sigma^*$ as well. 
For a finite set $S$, $|S|$ denotes the number of elements in $S$.
$\log$ denotes the logarithm function, typically to the base 2.
$O$ denotes the big-O notation. $o$ denotes the little-O
notation. $\omega$ denotes the little-$\Omega$
notation. See \cite{AroraBarak} for definitions.
In our work, $\poly$ denotes the
set of all univariate polynomials.
	A function $t : \mathbb{N} \to \mathbb{N}$ is \emph{time-constructible} if there exists a deterministic Turing machine that, on input $1^n$, outputs $t(n)$ in $\mathcal{O}(t(n))$ time. A prefix-free machine is a partial computable function $P:\{0,1\}^*\to\{0,1\}^*$ whose domain is prefix-free, i.e., no string in domain of $P$ is a prefix of another.
Let $M$ be an oracle Turing machine that computes a sequence $X \in \Sigma^\infty$ with oracle access to $A \in \Sigma^\infty$. For any $n \in \N$, the oracle use $u_n$ is the  the largest position of $A$ queried during the computation of $X \restr n$. $K(x)$ denotes the Kolmogorov complexity of the string $x \in \Sigma^*$ (See \cite{LiVitanyi}). In this work, quasi-polynomial-time denotes any time (constructible) bound $t(n)$ such that $t(n) = \omega(n^k)$ for all $k \in \N$. 

\section{Martin-L\"of randoms and Ku\v{c}era--G\'{a}cs Theorem}

In the theory of algorithmic randomness, a sequence is Martin-L\"of random if it avoids every effective test of randomness. 
We use the definition of randomness using betting strategies called martingales.
There are multiple equivalent approaches, for instance, one using Kolmogorov
complexity (see Downey and Hirschfeldt
\cite{Downey10} or Nies \cite{Nies2009}).
\begin{definition} [Martingales] \label{def:Mgales}
	A martingale is a function $d:\Sigma^* \to [0,\infty)$ such that
	$d(\lambda)\le 1$, and for every string $w$, we have
	$2 \cdot d(w)=d(w0)+d(w1)$. 
	
	A martingale succeeds on an infinite binary sequence $X$ if
	$\limsup_{n \to \infty} d(X \upharpoonright n) = \infty$. The
	\emph{success set} of a martingale $d$, denoted $S^\infty[d]$, is the
	set of infinite sequences over which $d$ succeeds. 
\end{definition}
Intuitively, the martingale value $d(w)$ represents the capital of a betting strategy after betting on the individual bits of $w$. This formulates a fair betting strategy in which the expected capital after any position equals the capital before the bet is placed.

A martingale $d: \Sigma^* \to [0,\infty)$ is said to be \emph{lower
	semicomputable} if there is a function $\hat{d}: \Sigma^* \times \N
\to [0,\infty) \cap \Q$ such that for each $w$ and $n$, we have
$\hat{d}(w,n) \le \hat{d}(w,n+1) \le d(w)$ and for every $w$, $\lim_{n
	\to \infty} \hat{d}(w,n) = d(w)$. 

\textbf{Universal martingale} : There is a universal lower
semicomputable martingale $d_0$ such that for any lower semicomputable
martingale $d$, there is a constant $c_d$ such that for all $w$, we
have $d(w) \le c_d \cdot d_0(w)$ (See \cite{Downey10}). Thus $d_0$ succeeds on all sequences that
any lower semicomputable martingale succeeds on.

\textbf{Martin-L\"of random} : An infinite binary sequence $X$ is said to be \emph{Martin-L\"of random} if a universal lower
semicomputable martingale $d_0$ \emph{does not} succeed on $X$. 

It is known that
the set of Martin-L\"of randoms has Lebesgue measure 1. Hence almost every sequence in $\Sigma^\infty$ is Martin-L\"of random. An extensive
body of work (see, for example, \cite{Downey10}, \cite{Nies2009})
shows that \emph{every} Martin-L\"of random exhibits many properties
which hold with probability 1.

\textbf{Ku\v{c}era--G\'{a}cs theorem} :
The following Ku\v{c}era--G\'{a}cs theorem is a fundamental property of
Martin-L\"of randomness. It asserts that every infinite sequence can
be effectively recovered from some Martin-L\"of random sequence.

\begin{theorem}[Ku\v{c}era \cite{Kucera84}, G\'{a}cs \cite{Gacs86}]
	\label{thm:KuceraGacs}
	For every sequence $X \in \Sigma^\infty$, there exists a Martin-L\"of
	random sequence $R \in \Sigma^\infty$ such that
	$
	X \leq_T R.
	$
	Moreover, the oracle use $u_n$ required to compute the first $n$ bits of
	$X$ from $R$ satisfies
	$
	u_n = n + o(n).
	$
\end{theorem}

Here, $X \leq_T R$ means that the sequence 
$X$ is computable by a Turing machine with oracle access to 
$R$.
Informally, the Ku\v{c}era--G\'{a}cs theorem states that any information
that can be encoded into an infinite binary sequence can be embedded into a Martin-L\"of random sequence in
such a way that it remains effectively recoverable \cite{BarmpaliasLewisPye18}.
Here, effective recoverability means computability via a Turing
reduction, with no restrictions on time or space.

\section{A Quasi-Polynomial time Ku\v{c}era--G\'{a}cs Theorem}

In this work, we study the polynomial-time variants of the Ku\v{c}era--G\'{a}cs theorem.  
We replace both randomness and reducibility with their polynomial-time analogues. Randomness is taken to mean polynomial-time randomness, i.e., a sequence on which no polynomial-time computable martingale succeeds. Additionally, the reduction from a random sequence to any given sequence is required to run in quasi-polynomial-time.

%TODO : Maybe mention PSPACE random -> poly-time redn does not work ?

We hence establish a quasi-polynomial-time analogue of the Ku\v{c}era--G\'{a}cs theorem. 
We show that
that for all sequences $X \in \Sigma^\infty$,
there exists a polynomial-time random $R \in \Sigma^\infty$ such that
$X$ can be computed from $R$, using a reduction that takes at most
quasi-polynomial-time. Additionally, we require only $n + o(n)$ bits
from $R$ to construct the first $n$ bits of $X$.

\subsection{Polynomial-time Randomness}

Polynomial-time randomness adapts classical notions of Martin-L\"of randomness to a computationally feasible polynomial-time setting. Rather than asking whether a sequence is random to any algorithm, this notion asks whether it appears random to any efficient polynomial-time algorithm. Defined via polynomial-time computable martingales, polynomial-time randomness formalizes the idea that no efficient betting strategy running in polynomial-time can gain anything but a finite advantage over a sequence. Hence, this captures the generalised notion of a polynomial-time test of randomness.

\begin{definition} [Juedes and Lutz \cite{JuedesLutz95}] 
	A martingale $d: \Sigma^{*} \to \Q$ is called  \exact $t(n)$-time
	computable \footnote{Technically, these martingales are called exact-$t(n)$-time computable martingales, 
		since their values are rationals. In this work, we omit the term 
		``exact'' for simplicity.}
	 if there exists a Turing machine $M$ such that, for every input $w \in \Sigma^n$, $M$ outputs $d(w)$ in time at most $t(n)$. 
\end{definition}

A martingale is called   \emph{polynomial-time
computable} if it is   $p(n)$-time computable for some polynomial $p(n)$.
%Here exact computability means that the martingale outputs are rationals. 
A sequence is polynomial-time random if no polynomial-time martingale
succeeds on it.

\begin{definition}  [Ambos-Spies, Neis, and Terwijn \cite{AmbosSpiesNeisTerwijn}] 
  The sequence $R \in \Sigma^\infty$ is \emph{polynomial-time random} if for
  all  \exact polynomial-time computable martingales $d$, we have that
  $R \not\in S^\infty[d]$.
\end{definition}

In this work, we use the notion of $\Delta$-Turing reducibility from Doty \cite{Doty08}, with $\Delta$ taken to be polynomial-time resource bounds. 

\begin{definition}[Doty \cite{Doty08}] \label{def:t(n)reduc}
	We say that \(X \in \Sigma^\infty\) is 
	\(t(n)\)-Turing reducible to \(R \in \Sigma^\infty\), 
	denoted \(X \le_{t(n)} R\), 
	if there exists an oracle Turing machine $M$ such that, 
	for all \(n \in \mathbb{N}\), 
	\(M^R(1^n)\) outputs \(X\upharpoonright n\) 
	in at most \(t(n)\) time.
\end{definition}

%\footnote
{ Definition \ref{def:t(n)reduc} differs from the more standard definition of a time-bounded Turing reduction used in complexity theory in which $X[n]$ is computed in time $t(|n|)$, instead of $X\upharpoonright n$ in time $t(n)$. We adopt the notion of time-bounded reducibility from Doty \cite{Doty08} in this work, as it aligns with the polynomial-time resource-bounded notions of martingales and randomness.}

\subsection{Quasi-polynomial-time universal
	martingale}

In order to construct a polynomial-time random $R$ from a given
sequence $X$, we diagonalise against a martingale $ {d}$
that is 
universal over all polynomial-time martingales. This means that
$ {d}$ succeeds on all sequences on which any polynomial-time
martingale succeeds.

Note especially that in the setting of polynomial-time martingales,
there is no universal polynomial-time martingale which succeeds on the
success set of all other polynomial-time martingales. But, using a
martingale combination technique from Ambos-Spies, Neis, and Terwijn \cite{AmbosSpiesNeisTerwijn}, for any time-constructible function
$t:\N \to \N$, we design a martingale that is universal over all
$t(n)$-time martingales and runs in time $t(n)\cdot n \log(n)$. 

\begin{restatable}{lemma}{universality}[Ambos-Spies, Neis, and Terwijn \cite{AmbosSpiesNeisTerwijn}] \label{lem:universality}
	For all time constructible functions $t : \N \to \N$,  there exists a
	 \exact $t(n)\cdot n \cdot \log(n)$ time martingale $ {d}$ such
	that for all  \exact $t(n)$-time martingales $d'$, $S^\infty[d'] \subseteq S^\infty[d]$.
%	for all $X \in \Sigma^\infty$, and for all
%	$n \in \N$, we have
%	$ {d} (X \restr n) \geq c \cdot d' (X \restr n)$.
\end{restatable}
\begin{proof}
	Let $\{M_i\}$ be a standard computable enumeration of all Turing
	machines. The idea of the proof is that we consider the $t(n)$-time
	martingales among the first $\log n$ machines to place bets on
	strings of length $n$. The remaining martingales are assumed to bet
	evenly on all strings of length $n$.
	
	We will therefore define the martingale, conditioned on the length
	of the input string $x$. When the length of the input string $x$
	satisfies $2^n \leq |x| < 2^{n+1}$, we consider the machines
	$M_1 \dots M_n$ to decide $\tilde{d}(x)$.
	
	The remaining machines $M_i$ for $i \geq n+1$ are assumed to bet
	\emph{evenly} on $x$. So, for $i \geq n+1$, we force the capital
	placed by $M_i$ on $x$ to be $1$.
	
	Therefore, given $x \in \Sigma^*$, such that
	$2^n \leq |x| < 2^{n+1}$, we define
	\begin{align*}
		{d}(x) &= \sum_{i =1}^n 2^{-i} \cdot M_i'(x) + \sum_{i=n+1} ^\infty {2^{-i}} \\
		&= 2^{- n } + \sum_{i =1}^n  2^{-i} \cdot M_i'(x). 
	\end{align*}
	
	We now define $M_i'(x)$. For $n \in N$, let $\D_n(x)$ be the
	set of machines in $\{M_1, \dots M_n\}$ such that for each
	$j \leq |x|$, $M(x \restr j)$ runs in time $t(j)$ and
	$M(\lambda) = 1$ and
	$2 \cdot M(x \restr j) = M((x\restr j).0) + M((x\restr j).1)$.
	
	\textbf{Case 1: $M_i\in \D_n(x)$}. The actual bets placed by
	martingale $M_i$ is taken into consideration for strings of
	length more than $2^i$. In that case, we need to ensure that
	the capital of $M_i$ at $x \restr 2^i$ is forced to be $1$.
	
	In this case, when $M_i(x \restr 2^i) > 0$, we define
	$M_i'(x)= (M_i(x)/M_i(x \restr 2^i))$. If
	$M_i(x \restr 2^i) = 0$, then it holds that $M_i(x) = 0$, so
	we just take $M_i'(x) = 0$.
	
	\textbf{Case 2: $M_i\not\in \D_n(x)$}. In this case, let $j <
	|x|$ be the least number such that $M_i(x\restr j)$ takes more
	than $t(j)$ time to run, or the martingale condition gets
	violated, that is
	$2 \cdot M_i(x \restr j) \neq M_i((x\restr j).0) + M_i((x\restr
	j).1)$. In this case, we freeze $M_i's$ capital at index $j$,
	and $M_i$ is forced to bet evenly after that. Note that we
	have to force $M_i$ to bet evenly up to $x \restr 2^i$.
	Therefore if $j < 2^i$, we take $M_i'(x) = 1$. Otherwise,
	$j \geq 2^i$ and $M_i'(x) = M_i'(x \restr j)$. Note that from
	the previous case, we have
	$M_i'(x \restr j) = (M_i(x \restr j)/M_i(x \restr 2^i))$.

	Consider any \exact $t(n)$-time  martingale $d'$. Let
	$M_k$ be any $t$-time machine such that $d'(w)=M_k(w)$ for
	every $w \in \Sigma^*$. Then, it follows from the construction
	of ${d}$ that for all $w \in \Sigma^*$ with $\lvert w \rvert >
	2^k$, \begin{align*} {d}(w) \geq \frac{2^{-k} }{d'(w
			\restr2^k)} \cdot d'(w).
	\end{align*}
	
	It follows that $S^\infty[d'] \subseteq S^\infty[d]$.
	
	The
	membership of any machine from $\{M_1,M_2,\dots M_n\}$ in
	$\D_n(w)$ is decidable in time $t(|w|) \cdot n \cdot |w|$.
	Since $n \leq \log(|w|)$, and $2^s \in \Q$, 
	$d$ is a \exact $t(n) \cdot n \cdot \log(n)$-time computable
	martingale.
	
\end{proof}

The following lemma formalizes the \emph{savings-account} trick. This shows that any martingale can be effectively transformed into one whose capital tends to infinity along exactly the same success set.  See Proposition~6.3.8 in Downey and Hirschfeldt~\cite{Downey10} for details of the construction. In addition, we record the time bound of the martingale obtained in the construction.

\begin{lemma}[Folklore] \label{lem:savingsaccount}
	Let $d'$ be a  \exact $t(n)$-time martingale. From $d'$ we
	can effectively define an \exact $n\cdot t(n)$-time \emph{(savings-account)} martingale $ {d}$ such that $
	S^\infty[ {d}] = S^\infty[d']
	$
	and for all $X \in \Sigma^\infty$ we have $
	\limsup_{n} d'(X \upharpoonright n) = \infty
	\text{ iff }
	\lim_{n}  {d}(X \upharpoonright n) = \infty .
	$
\end{lemma}

From Lemma \ref{lem:universality} and Lemma \ref{lem:savingsaccount}, it follows that there are quasi-polynomial-time savings-account
martingales which are universal over all polynomial-time martingales.
%The following lemma (Lemma \ref{lem:univ-polytime}) follows from Lemma
%\ref{lem:universality} and the observation that
%$n^k \cdot n\cdot\log(n) \in O(t'(n))$ for function $t'(n)$ such that
%$t' = \omega(n^k)$ for all $k \in \N$.

%(quasi)-polynomial-time universal
%martingale

\begin{restatable}{lemma}{univpolytime}\label{lem:univ-polytime}
	Let $t'(n)$ be a time-constructible function with $t' = \omega(n^k)$
	for all $k \in \mathbb{N}$. There exists a \exact $t'(n)$-time
	{savings-account} martingale $d$
	such that for every \exact polynomial-time computable martingale $d'$
	and every $X \in \Sigma^\infty$, if $X \in S^\infty[d']$ then
	$\lim_{n \to \infty} d(X \upharpoonright n) = \infty$.
\end{restatable}

\begin{proof}
	Let $t'(n)$
	be any time constructible function such that $t' = \omega(n^k)$ for all $k \in \N$. Take $t''(n) = t'(n)/(n^2 \log(n))$. \footnote{More precisely, $t''(n) = \lfloor t'(n)/(n^2 \log(n)) \rfloor$.} Note that $t''(n) = \omega(n^k)$ for all $k \in \N$.  All \exact polynomial-time martingales are also \exact $t''(n)$-time martingales. So from Lemma \ref{lem:universality}, we get a \exact  $t'(n)/n$-time martingale $ {d''}$ such
	that for all \exact polynomial-time computable martingales $d'$, we have $S^\infty[d'] \subseteq S^\infty[d'']$. From Lemma \ref{lem:savingsaccount}, we get a $t'(n)$-time savings-account martingale $d$ such that for all $X \in S^\infty[d'']$ (and thereby all $X \in S^\infty[d']$), we have $\lim_{n \to \infty} d(X \upharpoonright n) = \infty$.
\end{proof}

\begin{remark} \label{remark}
	From this point on in the paper, we use $d$ to denote a quasi-polynomial-time universal savings-account martingale. Note  that a sequence $X \in \Sigma^\infty$ is polynomial-time random if $\liminf_{n \to \infty} d(X \restr n) < \infty$.
\end{remark}

\subsection{$N$ guides the way}

Given a sequence $X$, we use the bits from $X$ to diagonalise against the martingale ${d}$ (see Remark \ref{remark}) to construct $R$, such that ${d}$ does not succeed on $R$. This ensures that $R$ is polynomial-time random.

We do a careful stage-wise construction in which at stage $i$, the next $i$ bits of $X$ are used to fix the corresponding $\ell_i$ bits (Definition \ref{def:param}) in $R$. This technique of encoding $i$ bits at once (Gacs \cite{Gacs86}, Merkle and Mihailovic \cite{MerkleMihailovic})  is done to reduce the redundancy of $R$ (number of bits of $R$ needed to obtain $n$ bits of $X$) to $n + o(n)$.

We choose $\delta_i's$ (Definition \ref{def:param}) so that $\prod_{i =1}^\infty \delta_i < \infty$. We ensure that the capital gained by $R$ at the end of stage $i$ is less than $\delta_i$ times the capital at the preceding stage. We call such an extension to $R$ a \emph{losing path}.

\begin{definition} \label{def:param}
	For each \(i\in\mathbb{N}\), define
	\(  \delta_i \;=\; 1 + {i^{-2}}, 
	\quad \text{and}\quad 
	\ell_i \;= \;	\ceil{
	i + \log (1+i^{2})}.
	\)
\end{definition}

Let $w \in \Sigma^*$ to be the bits of $R$ fixed at the previous stage (till stage $i-1$).
For any $y \in \Sigma^{\leq \ell_i} $, the parameter $N(w,y,\ell_i,\delta_i)$ gives a \emph{lower bound} on the number of $\ell_i$-length extensions of $y$ (say to $y.y'$), for which $d(wyy') < d(w) \cdot \delta_i$.

\begin{restatable}{lemma}{KolmogorovIneqbinSearch}\label{lem:KolmogorovIneqbinSearch} Given an
  $w \in \Sigma^*$, $\ell,n \in \N$ with $n \leq \ell$,
  $y \in \Sigma^n$, and $\delta > 1$,
  the number of $y' \in \Sigma^{\ell - n}$ such that
  $d(wyy') < d(w) \cdot \delta$ is greater than or equal to
$$ N(w,y,\ell,\delta) = \frac{2^\ell}{2^{n}} \bigg[1 - \frac{d(wy)}{d(w)  \cdot \delta}\bigg].$$ 
Moreover,  for any $y \in \Sigma^{<\ell}$
we have $N(w,y.0,\ell,\delta) + N(w,y.1,\ell,\delta) = N(w,y,\ell,\delta).$ 
\end{restatable}

\begin{proof}
	We have that $\sum_{y' \in \Sigma^{\ell - n}} d(wyy') \leq 2^{\ell - n} \cdot d(wy)$.
	
	From this it follows that the number of $y' \in \Sigma^{\ell - n}$ such that
	$d(wyy') \geq d(w)  \cdot \delta$ is lesser than or equal to $ 2^{\ell - n} \cdot \frac{d(wy)}{d(w)  \cdot \delta}$. 
	
	Therefore the number of $y' \in \Sigma^{\ell - n}$ such that
	$d(wyy') < d(w) \cdot \delta$ is greater than or equal to ${2^{\ell - n}} \bigg[1 - \frac{d(wy)}{d(w)  \cdot \delta}\bigg]$. 
	
	Moreover, $N(w,y.0,\ell,\delta) + N(w,y.1,\ell,\delta) = \frac{2^\ell}{2^{n}} \bigg[2 - \frac{d(w.y.0) + d(w.y.1)}{d(w)  \cdot \delta}\bigg] = \frac{2^\ell}{2^{n}} \bigg[2 - \frac{2 \cdot d(w.y)}{d(w)  \cdot \delta}\bigg]$ = $N(w,y,\ell,\delta)$.
\end{proof}

At stage $i$, we use the next $i$ bits of $X$, and choose a losing path for $R$. We use $y$ to denote the extension to $R$ made at stage $i$ so far.
Note that at the beginning of stage $i$, taking $y = \lambda$, we have 
$N(w, \lambda, \ell_i,\delta_i) = 2^{\ell_i} [1 - \frac{1}{\delta_i}] \geq 2^i$.
Therefore, there are at least $2^i$ losing paths among $\ell_i$ extensions of $R$, and therefore it is possible to encode the next $i$ bits of $X$ into a unique $\ell_i$-length losing path in $R$.

Let $val(x)$ be the numerical value corresponding to the next $i$ bits of $X$ (we add $1$ to keep this value between $1$ and $2^i$) encoded at stage $i$.
A naive approach to do the encoding would be to find 
assign the $val(x)^{th}$ losing path in $R$.
However, while decoding $X$ from $R$, we cannot go over all the $2^{\ell_i}$ choices of $R$, to determine if its a losing path to recover $x$. Such a brute-force approach would take exponential time.

We therefore use a binary-search approach using the values of $N(.)$, for the encoding. This ensures that the decoding can be performed in quasi-polynomial-time. At the start of stage $i$, we have the value of $N(.)$ to be at least $2^i$.
As $N$ is additive, we can compute locally how appending a $0$ or $1$ to the current choice $y$ affects  the distribution of the losing paths for $R$. We make the choice of extending $R$ with a $0$ or $1$, based on this distribution, and $val(x)$.  (see Subsection \ref{subsec:KGPtime}).

{\bf Notation.} Given an $w \in \Sigma^*$, $y \in \Sigma^n$,
$\ell \geq n$ and $\delta > 1$, $N(w,y,\ell,\delta)$ denotes a \emph{lower bound} on the number of
$y' \in \Sigma^{\ell - n}$ such that $d(wyy') < d(w) \cdot \delta$ obtained from Lemma \ref{lem:KolmogorovIneqbinSearch}, in the rest of this paper.

%Note that for any $x \in \Sigma^{<\ell}$, we have $N(w,x.0,\ell,\delta) + N(w,x.1,\ell,\delta) = N(w,x,\ell,\delta)$. Also note that $N(w, \lambda, \ell,\delta) = 2^{\ell} [1 - \frac{1}{\delta}]$. Given any $i \in \N$, plugging in $\ell = i + \log(1 + i^{2})$ and $\delta = 1 + i^{-2}$, we get that $N(w, \lambda, \ell,\delta) = 2^i$. 

\subsection{Quasi-Polynomial time Ku\v{c}era--G\'{a}cs Theorem} \label{subsec:KGPtime}

We now show a quasi-polynomial-time analogue of the Ku\v{c}era
G\'{a}cs Theorem. The theorem states that for any sequence
$X \in \Sigma^\infty$, there exists a polynomial-time random sequence
$R$, such that $X$ can be recovered from $R$ using a quasi-polynomial-time reduction.

\begin{restatable}{theorem}{PtimeKuceraGacs}\label{thm:quasiPolynKG}
	For all $X \in \Sigma^\infty$, there exists an $R \in \Sigma^\infty$
	such that $R$ is polynomial-time random and $X \leq_{t(n)} R$ where
	$t(n) : \N \to \N$
	is any time-constructible function such that $t(n) = \omega(n^k)$ for all $k \in \N$.
	
	Additionally, to obtain the first $n$ bits of $X$, we only need the first $n + o(n)$ bits of $R$.
	%	$t(n)$ is such that $t(n) = \omega(n^k)$ for all $k \in \N$. 
\end{restatable}

\begin{proof}

\subsubsection{Encoding}

  We first show how to build $R$ given $X$
  stage-wise. At each stage $i$, using the next $i$ bits of $X$, we
  extend the string $R_{i-1}$ built so far to $R_i$ by $\ell_i$ bits
  such that $ {d}( R_{i})\leq \delta_i \cdot  {d}(R_{i-1})$. We provide the details below. See Algorithm \ref{alg:extendR} for the formal construction. 

\medskip

 At stage $i$, let
$x \in \Sigma^i$ be the next $i$ bits from $X$. Let $val(x)$ be the numerical
value corresponding to the binary string $x$ (Note : We add 1 to the numerical value to keep the index strictly between $1$ and $2^i$, so $val(0^{i}) = 0 + 1 = 1$ and so on).
 We use $val(x)$ and $N$ (from Lemma \ref{lem:KolmogorovIneqbinSearch})  
to  extend
$w = R_{i-1}$ to one among the (at least) $2^i$ choices in $R_{i}$ using  binary search.

 At the end of stage $i$, we extend $R$ by the string $y$, which we append bitwise, iterating over a variable $j$ going 
 from $0$ to $\ell_i - 1$. $y$ is initialised to the empty string $\lambda$.
We use an auxiliary variable $a_j$ to denote
 the cumulative value of $N$ of the nodes to the left (lexicographically smaller) of the current node $y$.  $a_0$ is initialised to $0$. 

At iteration $j$, we have fixed the first $j$ bits of $y$, we need to fix the $j+1 ^ {th}$ bit of $y$. 
We compute the value of $N(w,y.0,\ell_i,\delta_i)$. 
If $val(x) \leq a_j + N(w,y.0,\ell_i,\delta_i)$, we append $y$ with $0$, and the value of $a_{j+1}$ is the same as $a_j$. Otherwise, we append $y$ with a $1$ and $a_{j+1}$ becomes $a_j + N(w,y.0,\ell_i,\delta_i)$.

\begin{algorithm}[]
	\caption{Extending $R_{i-1}$ to $R_i$}
	\label{alg:extendR}
	\begin{algorithmic}[1]
		\Require Stage $i$, input $x \in \Sigma^i$, previous string $w = R_{i-1}$
		\State Initialize $a_0 \gets 0$
		\State Initialize $y \gets \lambda$
		\For{$j = 0$ to $\ell_i - 1$}
		\If{$val(x) \le a_j +  N(w,y.0,\ell_i,\delta_i)$}
		\State $y \gets y.0$
		\State $a_{j+1} \gets a_j$
		\Else
		\State $y \gets y.1$
		\State $a_{j+1} \gets a_j + N(w,y.0,\ell_i,\delta_i)$
		\EndIf
		\EndFor
		\State $R_i \gets w.y$
	\end{algorithmic}
\end{algorithm}

So, corresponding to each stage $i$, we get an increasing sequence of
finite prefixes \(R_0 \prec R_1 \prec R_2 \prec \cdots\).
Define $R = \lim_{i \to \infty} R_i$.
\begin{restatable}{claim}{loopinvariant}\label{claim:loopinvariant}
Throughout the iteration, the following invariants are maintained : 
\begin{enumerate}
	\item $a_j < val(x) \leq a_j + N(w,y,\ell_i,\delta_i).$
	\item $N(w,y,\ell_i,\delta_i) > 0$
\end{enumerate}
\end{restatable}

\begin{proof}
	Invariant 2 follows from Invariant 1.
	
	Invariant 1 is maintained at the start of the iteration as 
	$N(w, \lambda, \ell_i,\delta_i) \geq 2^i$ and $0 < val(x) \leq 2^i$. Suppose the invariant is maintained, till iteration $j$.  
	
	If $val(x) \le a_j +  N(w,y.0,\ell_i,\delta_i)$, then $a_{j+1} = a_j$, and $y \gets y.0$. So $val(x) \leq a_{j+1} + N(w,y.0,\ell_i,\delta_i)$, and $val(x) > a_{j+1}$ as $val(x) > a_j$. So invariant 1 holds in this case. 
	
	Otherwise, $val(x) > a_j +  N(w,y.0,\ell_i,\delta_i)$, then $a_{j+1} = a_j + N(w,y.0,\ell_i,\delta_i)$, and so $val(x) > a_{j+1}$ holds.  As $val(x) \leq a_j + N(w,y,\ell_i,\delta_i)$, from the additive property of $N$, we have $val(x) \leq a_{j} + N(w,y.0,\ell_i,\delta_i) + N(w,y.1,\ell_i,\delta_i)$. So, we have $val(x) \leq a_{j+1} + N(w,y.1,\ell_i,\delta_i)$. So invariant 1 holds in this case as well. 
\end{proof}

Using Claim \ref{claim:loopinvariant}, we show that we are guaranteed a losing path at the end of every stage.

\begin{restatable}{claim}{loopend}\label{claim:loopend1}
	For any $i \in \N$, for the $y$ obtained at the end of stage $i$ in Algorithm \ref{alg:extendR}, we have $d(w.y) < \delta_i \cdot d(w)$. 
\end{restatable} 
\begin{proof}
	From Invariant 2, at the end of stage $i$, $N(w,y,\ell_i,\delta_i) > 0$. From the definition of $N$ (Lemma \ref{lem:KolmogorovIneqbinSearch}), this holds only if $d(w.y) < \delta_i \cdot d(w)$.
\end{proof}

Since \(\sum_i1/i^2<\infty\), so we have
\(\displaystyle\prod_{i=1}^\infty \delta_i<\infty\). So
$\liminf_{n \to \infty}  {d}(R \restriction n) \leq \prod_{i =
  1}^{\infty} \delta_i < \infty$. It follows that $R$ is polynomial-time random (See Remark \ref{remark}).

%% ==================================================
%% SN: Elaborate the importance
%% ==================================================

\subsubsection{Decoding}

Algorithm~\ref{alg:decodeX} describes the decoding procedure. At stage $i$,
the algorithm recovers the next $i$ bits of $X$, using $R$.  This is done by running Algorithm~\ref{alg:extendR} in reverse. 

Let $w$ denote the
prefix of $R$ read up to stage $i-1$. We initialize $y$ to $\lambda$ and
iteratively update it using the bits of $R$ read during stage $i$. We
maintain the set of possible values of $val(x)$ as an interval $(start,end]$,
initially set to $(0,2^i]$.
At each step, the interval is updated according to the current bit of $R$ being read. If the bit is $0$, then by Algorithm~\ref{alg:extendR} it must hold that
$val(x) \le start + N(w,y.0,\ell_i,\delta_i)$. Otherwise,
$val(x) > start + N(w,y.0,\ell_i,\delta_i)$. Thus, at each iteration, the
interval $(start,end]$ containing $val(x)$ is narrowed accordingly. We show that after
$\ell_i$ iterations, the resulting interval uniquely determines $x$.

\begin{restatable}{claim}{loopendtwo}\label{claim:loopend2}
	For any $i \in \N$, at the end of stage $i$ in Algorithm \ref{alg:decodeX}, there is a unique $x$ such that $val(x) \in (start, end]$.
\end{restatable} 
\begin{proof}
	Assume that two $x,x' \in \Sigma^i$ (with $val(x) \leq val(x')$) are in $(start,end]$ at the end of stage $i$ in Algorithm \ref{alg:decodeX}.
	It follows that $start < val(x) \leq val(x') \leq start + N(w,y,\ell_i,\delta_i)$. Since $|y| = \ell_i$, from the definition of $N$ (Lemma \ref{lem:KolmogorovIneqbinSearch}), it holds that $N(w,y,\ell_i,\delta_i) \leq 1$. As $val(x)$ and $val(x')$ are integers, we have that $val(x) = val(x')$, and so $x = x'$.
\end{proof}

\begin{algorithm}[]
	\caption{Using $R_i$ to compute $x$}
	\label{alg:decodeX}
	Stage $i$:
	\begin{algorithmic}[1]
		
		\State  Initialize $k_{i-1} \;\gets\;\sum_{j=1}^{i-1}\ell_j $
		\State  Initialize $w \;\gets R[ :k_{i-1}]$ \Comment{The first $k_{i-1}$ bits of $R$}
		\State Initialize $start \gets 0$, $end \gets 2^i$ \Comment{$ val(x) \in (start, end ]$}
		\State Initialize $y \gets \lambda$
		\For{$j = 0$ to $\ell_i-1$}
		\If{$R[k_{i-1} + j] = 0$} \Comment{$ val(x) \le start +  N(w,y.0,\ell_i,\delta_i)$}
		\State $end \gets start +  N(w,y.0,\ell_i,\delta_i)$
		\State $y \gets y.0$
		\Else \Comment{$ val(x) > start +  N(w,y.0,\ell_i,\delta_i)$}
		\State $start \gets start +  N(w,y.0,\ell_i,\delta_i)$
		\State $y \gets y.1$
		\EndIf
		\EndFor
		\State Find $x$ such that $val(x) \in (start, end]$
	\end{algorithmic}
\end{algorithm}
 We now analyse the redundancy of the reduction, that is the number of oracle bits of $R$ used to compute the first $n$ bits of $X$.
	 
	 \textbf{Redundancy.} At the end of the
         $i^{th}$ stage, we obtain the first
         $n_i = \sum_{j = 1}^{i} j = \frac{i (i+1)}{2}$ bits of $X$.
         To obtain this, the number of bits of $R$ used,
         \begin{align*}
           m_i \;&= \sum_{j=1}^i \ell_j
                     = \sum_{j=1}^i  j + \log (1+j^{2})
                     = \frac{i (i+1)}{2} + \sum_{j=1}^i \log
                       (1+j^{2}) \\ 
                     &\leq n_i + 2i\log i \leq n_i + \sqrt{2n_i}
                       \log(2n_i) 
                     = n_i + o(n_i).
                 \end{align*}  
         Note that to obtain the bits between $n_{i}$ and $n_{i+1}$,
         we still need to query $m_{i+1}$ bits of $R$. But we have
         $m_{i+1} \leq n_{i+1} + \sqrt{2n_{i+1}} \log(2n_{i+1})$.
         Using $n_{i+1} = n_i + i + 1$, and $\sqrt{2n_i} > i$, we have
         $m_{i+1} \leq n_{i} + 2 \sqrt{n_i} \log(n_i)$ which is still
         $n_{i} + o(n_i)$.

         \textbf{Time Complexity.} :
         Let $t(n) : \N \to \N$
         be any function such that $t(n) = \omega(n^k)$ for all $k \in \N$.
  		 Take $t'(n) = t(n/2)/n$. \footnote{More precisely, take $t'(n) = \lfloor t(\lfloor n/2 \rfloor)/n \rfloor$.} Note that $t'(n) = \omega(n^k)$ for all $k \in \N$. We take $t'(n)$ to be the running time of the martingale $d$ (See Lemma \ref{lem:univ-polytime}).
  		  	       
         At stage \(i\), we evaluate \( {d}\) on \(\ell_i\)
         strings, with each evaluation taking time at most
         \(t'(m_i)\). So the total computation time after
         stage $i$ ends is less than $m_i \cdot t'(m_i)$.
         For large enough $n_i$, we have $m_i \leq 2\cdot n_i$. 
         So to compute the first $n$ bits of $X$, the total time
         taken is at most $\mathcal{O}(n \cdot t'(2 \cdot n)) = \mathcal{O}(t(n))$.
\end{proof}

\section{Polynomial-time dimension and Polynomial-time reductions}

The constructive dimension of a sequence represents the density of information in the sequence.

\begin{theorem}[Lutz \cite{Lutz03}, Mayordomo \cite{Mayordomo02}]
	For all $X \in \Sigma^\infty$, 
	$
	\cdim(X) = \liminf_{n \to \infty} \frac{K(X \restr n)}{n}.
	$
\end{theorem}

Doty \cite{Doty08} shows that the constructive dimension of a sequence can be characterised by the rate of oracle use in the 
Ku\v{c}era--G\'{a}cs reduction. This means that the amount of random bits queried from the source random $R$ is proportional to the rate of randomness in the target sequence $X$.

\begin{theorem}[Doty \cite{Doty08}] \label{thm:Doty}
	For all $X \in \Sigma^\infty$, there exists a Martin-Löf random $R \in \Sigma^\infty$ such that
	$X \leq_T R$ via $M$ with oracle use $u_n$ such that
	\(
	\cdim(X) = \liminf_{n \to \infty} \frac{u_n}{n}.
	\)
\end{theorem}

Doty further shows an analogous result for resource bounded dimension of a sequence $X$ at
computable, polynomial-space settings. However, such a
characterization for polynomial-time dimension remained open.

\textbf{Polynomial-time density of information} :
At the resource bound of polynomial-time, there are many methods to quantify density of information.
The polynomial-time lower oracle use rate ($\rho^-_{\mathsf{poly}}$) of a sequence captures the least rate of oracle information needed to compute the sequence efficiently from any (not necessarily random) target source sequence.
$\Kpoly$, on the other hand, is the lower polynomial-time Kolmogorov complexity rate of a sequence. Polynomial time density of information can also be quantified by using betting algorithms, called polynomial-time $s$-gales. The corresponding notion is called polynomial-time dimension ($\cdimp$).

We show unconditionally that the
notions $\Kpoly$ and $\rho^-_{\mathsf{poly}}$ are equivalent (Theorem
\ref{thm:eqvlRhoKpoly}).
 Akhil, Nandakumar,
Pulari, and Sarma \cite{OWFDimension} show that $\cdimp$ and $\Kpoly$
are different quantities, under the cryptographic assumption that
one-way functions exist. Using this result, we show that
$\rho^-_{\mathsf{poly}}$ and $\cdimp$ are different, if one-way
functions exist.

So for every sequence $X$,
$\rho^-_{\mathsf{poly}}(X) = \Kpoly(X)$. But if there are
cryptographic one-way functions, then there are sequences $Y$ for
which $\rho^-_{\mathsf{poly}} (Y)< \cdimp(Y)$. Hence Doty's characterization
\emph{fails} in the polynomial-time setting, if one-way functions
exist.

Further, we extend the quasi-polynomial-time Ku\v{c}era--G\'{a}cs to
show that the number of bits of the polynomial-time random $R$ used to
construct the first $n$ bits of $X$ is proportional to the $\Kpoly$
dimension of a sequence (Theorem \ref{thm:KuceraGacsPolyExtended}).

\subsection{Polynomial-time Kolmogorov complexity rate : $\Kpoly$}

The $t$-time-bounded Kolmogorov complexity of a string $x \in \Sigma^*$,
denoted $K^t(x)$, is the length of the shortest program that outputs
$x$ within time $t(|x|)$. We follow the convention of Sipser \cite{Sipser83}, and keep $t$ to be a function.

\begin{definition} [Sipser \cite{Sipser83}]  \label{def:Kt}
  Let \( t \) be a time-constructible function, and let
  \( x \in \Sigma^* \) be a finite binary string. The
  \emph{\( t \)-time bounded Kolmogorov complexity} of \( x \) is
	
	\(
          K^t(x) = \min\{ |\Pi| \mid \Pi \in \Sigma^* \text{ and }
          U(\Pi) = x \text{ in } t(|x|) \text{ steps} \},
	\)
	where \( U \) is any fixed universal prefix-free Turing machine. 
\end{definition}

%$\mathcal{K}_{\mathsf{poly}}$ is the infimum of lower rate of $K^t$ of a sequence, over all polynomials $t$.

\(\mathcal{K}_{\mathsf{poly}}(X)\) is the infimum, over all
polynomial time bounds \(t\), of the lower
asymptotic rate of 
\(K^t(X\upharpoonright n)\).

\begin{definition} [Hitchcock and Vinodchandran \cite{dercjournal}] 
	For any \( X \in \Sigma^\infty \),	
	
	\(
          \mathcal{K}_{\mathsf{poly}}(X) = \inf\limits_{t \in \text{poly}}
          \liminf\limits_{n \to \infty} \frac{K^t(X \upharpoonright n)}{n}.
	\)
	
\end{definition}
\subsection{Polynomial-time oracle use rate : $\rho^-_{\mathsf{poly}}(X)$}

Given sequences $X, Y$, we say that $X \leq_P Y$ if there exists a Turing machine (OTM) $M$ with oracle access to $Y$ that produces $X$, taking at most $t(n) \in \poly$ time to produce the
first $n$ bits of $X$.

$\rho^-_M(X, Y)$
is the lower limiting ratio of the number of
bits of $Y$ queried by $M$ to produce the first $n$ bits of $X$.

\begin{definition} [Doty \cite{Doty08}]
	Let \( X, Y \in \Sigma^\infty \) and \( M \in \text{OTM} \)
        such that \( X \leq_P Y \) via \( M \).   Define 
	$	\rho^-_M(X, Y) = \liminf_{n \to \infty} \frac{\#(X
            \upharpoonright n, M^Y)}{n}  
$,	
	where \( \#(X \upharpoonright n, M^Y) \) is the number of bits
        of \( R \) queried by \( M \) when computing the string \( X
        \upharpoonright n \). \footnote{If we instead define \( \#(X
        \upharpoonright n, M^Y) \) to be the index of the rightmost
        bit of \( Y \) queried by \( M \) when computing \( X
        \upharpoonright n \), all results of the present paper still
        hold.}
\end{definition}

The lower polynomial-time decompression ratio of a sequence
$X \in \Sigma^\infty$ is the optimal value of $\rho^-_M(X, Y)$ over
all such polynomial-time machines $M$ and oracles $Y$.

\begin{definition}[Doty \cite{Doty08}]
  The \emph{lower polynomial-time decompression ratio} of a sequence
  \( X \in \Sigma^\infty \) is defined as
	\(
          \rho^-_{\mathsf{poly}}(X) = \inf\limits_{\substack{Y \in
              \Sigma^\infty \\ M \in \text{OTM}}} \left\{ \rho^-_M(X,
            Y) \;\middle\vert\; X \leq_P^{ } Y \text{ via } M
          \right\}.
	\)

\end{definition}

\subsection{Equivalence between $\Kpoly$ and $\rho^-_{\mathsf{poly}}$}
The following lemma shows that $\Kpoly$ is less than or
equal to $\rho^-_{\mathsf{poly}}(X)$.  The proof follows by the same argument as in Doty \cite[Lemma 4.4]{Doty08}. Suppose there is a polynomial-time machine $M$ that
infinitely often queries less than $s\cdot n$ oracle bits to produce $X\restr n$.
Then we can use the machine $M$ as well as the oracle use to produce
$X\restr n$ in polynomial-time. 

Doty showed that
$\dim_\Delta(X) \leq \rho^-_{\Delta}(X)$ for resource
bounds $\Delta$ including computable time and
polynomial space. His proof relies on
Kolmogorov complexity characterizations of
dimension $\dim_\Delta(X)$ that hold for these resource bounds.
For polynomial-time, such a characterization
does not hold (under standard cryptographic
assumptions \cite{OWFDimension}), so Doty did
not obtain the corresponding bound for
$\cdimp$.
However, we observe that the same argument
directly yields the following bound for
$\mathcal{K}_{\mathsf{poly}}(X)$.
We include the proof for completeness.

\begin{restatable}{lemma}{DotyLem}[Doty \cite{Doty08}] \label{lem:DotyLem} 
For any
  \( X \in \Sigma^\infty \),
  $\mathcal{K}_{\mathsf{poly}}(X) \leq \rho^-_{\mathsf{poly}}(X).$
\end{restatable}

\begin{proof}
	Let $s > \rho^-_{\mathsf{poly}}(X)$. There exists a
	$Y \in \Sigma^\infty$, a $t \in \poly$ and a $t(n)$-time machine
	$M \in $ OTM such that $X \leq_{t(n)} Y$ via $M$ such that for
	infinitely many $n \in \N$,
	$\#(X \upharpoonright n, M^Y) < s\cdot n$. We show that
	$K^t(X \upharpoonright n) \leq s\cdot n$, which establishes the
	lemma.
	
	Let \( p_n \in \{0,1\}^{\#(X \upharpoonright n, M^P)} \) be the
	oracle bits of \( P \) queried by \( M \) on input \( n \), in the
	order in which they are queried. We assume without loss of
	generality that the oracle queries made by $M$ are distinct. This is
	because $M$ can store the result of queries to $Y$, using an extra
	linear amount of space.
	
	We encode the information in a stagewise manner. To specify this
	encoding, we fix the convention that $w_0, w_1, \dots$ is the
	standard enumeration of the set of finite strings. Let $\text{enc}:
	\Sigma^* \to \Sigma^*$ be a prefix-free encoding of binary strings -
	\emph{i.e.} the set $\{\text{enc}(w) \mid w \in \Sigma^*\}$ is a
	prefix-free set. There are encodings which ensure that
	$|\text{enc}(w)| \le |w| + 2 \log(|w|)+2$ (for example, see Li and
	Vitanyi, Chapter 3. \cite{LiVitanyi}).
	
	Consider the program
	$\pi_n = \pi_M.\text{enc}(w_n).\text{enc}(p_n)$, where $\pi_m$ is
	the prefix free description of the machine $M$. We have that
	$U(\pi_n) = X \restr n$, where $U$ is the machine that first decodes
	$M,n,p_n$ and simulates the run of $M^Y$ on $X \restr n$ using the
	$p_n$ as the answer to the oracle queries.
	
	We have that
	\begin{align*}
		\Kpoly(X) \leq
		\liminf_{n \to \infty} \frac{K^t(X \restr n)}{n} \leq
		\liminf_{n \to \infty}
		\frac{|p_n|+\log(|p_n|)+|w_n|+\log(|w_n|)+|\pi_M|}{n}.
	\end{align*}
	Note that $|w_n| = O(\log n)$ and $|\pi_M|$ is a constant. Thus, we have
	\begin{align*}
		\Kpoly(X) \leq
		\liminf_{n \to \infty}
		\frac{|p_n|+2\log(|p_n|)}{n} \leq
		\liminf_{n \to \infty} \frac{sn+2 \log(sn)}{n}
		= s.
	\end{align*}
\end{proof}

We show the converse as well, \emph{i.e.} $\rho^-_{\mathsf{poly}}(X)\leq \Kpoly$, so the equality also holds in polynomial-time. We sketch the idea below.

Suppose that for infinitely many $n$, there is a short description of $X\restr n$ of size $s\cdot n+O(1)$. Choose a subsequence $n_1,n_2,\dots$ such that $\sum_{i=1}^k n_i=o(n_{k+1})$ for all $k$. In the unrestricted-time setting, concatenating the short descriptions of $X\restr n_i$ into an oracle $Y$ suffices, since the overhead for storing previous descriptions is negligible.

In the polynomial-time setting, however, this naive encoding may fail: to obtain the $(n_i+1)^{th}$ bit, the machine may need to read another $s\cdot n_{i+1}$ bits of $Y$, which can be super-polynomial if, say, $n_{i+1}=2^{n_i}$. So concatenation alone does not work.

We therefore use a hybrid encoding. Before the description $\pi_i$ of $X\restr n_i$, we hardcode the first $\sqrt{n_i}$ bits of $X[n_{i-1}:n_i]$, denoted $v_i$, into $Y$. To let the machine detect where this hardcoded part ends without scanning too much of $Y$, we insert a $0/1$ flag after each $2^j$-th bit of $v_i$, for each $j<\log(\sqrt{n_i})$.
This encoding guarantees that the machine outputs $X\restr n$ in polynomial-time for all $n$, including the intermediate indices $n_i<n<n_{i+1}$.

\begin{restatable}{theorem}{eqvlRhoKpoly} \label{thm:eqvlRhoKpoly} 
  For any \( X \in \Sigma^\infty \),
$ \rho^-_{\mathsf{poly}}(X) = \mathcal{K}_{\mathsf{poly}}(X) $	.
\end{restatable}

\begin{proof}
	From Lemma \ref{lem:DotyLem}, it suffices to show that
	$ \rho^-_{\mathsf{poly}}(X) \leq \mathcal{K}_{\mathsf{poly}}(X)$.
	
	Take any $s > \Kpoly(X)$. There exists a $t \in \poly$ and
	infinitely many $\{n_i\}_{i \in \N}$ such that for all $i \in \N$,
	$K^t(X \restr n_i) \leq s\cdot n_i$.
	
	Given $i \in \N$, define $m_i$ as follows $m_0 = 0$ and for
	$i \geq 1$, $m_{i} = n_j$ for the least $j \in \N$ such that
	$m_{i}\geq (\sum_{k = 1}^{i - 1} m_k)^2$. Let $\pi_i$ be the prefix
	free program such that $U(\pi_i) = X \restriction m_i$. We have that
	$|\pi_i| \leq s\cdot m_i$.
	
	At stage $i \in \N$, let $\ell_i = \frac{1}{2} \log m_i$. Define
	$v_i$ = $0.u_i^1.0.u_i^2. \dots u_i^{l_i}.1$. Here
	$u_i^j = X[m_{i-1} + 2^{j - 1} : m_{i-1} + 2^j]$. 
	Note that $|v_i| \leq 2 \sqrt{m_i} $.
	
	Finally define
	$Y = v_1.\pi_1.v_2.\pi_2.v_3.\pi_3 \dots$.
	
	We can see that $X \leq_{t'(n)} Y$ via the machine given in
	Algorithm \ref{algm:1}. To produce $X \restr n$, $M$ uses time at
	most $t'(n) = t(n) + n + t(n^2)$, which is polynomial-time as $t(n)$
	is polynomial-time.
	
	For all $i \in \N$, we have
	$ \#(X \upharpoonright m_i, M^Y) \leq \sum_{j = 1}^{i} |\pi_j| + |v_j|$.
	
	As $|\pi_j| \leq s \cdot m_j$, and $|v_i| \leq 2 \sqrt{m_i}$, we have $\#(X \upharpoonright m_i, M^Y) \leq  \sum_{j = 1}^{i} s \cdot m_j + 2  \sqrt{m_j} \leq s\cdot m_i + 2 \sqrt{m_i} + 3.\sum_{j = 1}^{i-1} m_j $.  
	
	By the choice of $m_i$,  $\sqrt{m_{i}}\geq \sum_{j = 1}^{i - 1} m_j$.  So, we have $ \#(X \upharpoonright m_i, M^Y) \leq s \cdot m_i + 5 \sqrt{m_i}$. From
	this it follows that $\rho^-_{\mathsf{poly}}(X) \leq s$.
	
	\begin{algorithm} [H]
		\caption{Machine $M$ on input $n$} \label{algm:1}
		\begin{algorithmic}[1]
			\State Set $x \gets \lambda$
			\For{each stage $i$ until $X \restr n$ is recovered}
			\State Set $x' \gets \lambda$
			\For{$j \in \mathbb{N}$}
			\State Read the next bit
			\If{bit is $0$}
			\State Copy next $2^j$ bits from $Y$ and append to $x'$
			\ElsIf{bit is $1$}
			\State Read $\pi_j$
			\State Set $x \gets U(\pi_j)$ \hfill \textit{\*Run $\pi_j$, store the output to $x$.}
			\If{$|x \cdot x'| < n$}
			\State \textbf{continue} to next stage $i$
			\EndIf
			\EndIf
			\If{$|x \cdot x'| \geq n$}
			\State \Return $x \cdot x' \restr n$
			\EndIf
			\EndFor
			\EndFor
		\end{algorithmic}
	\end{algorithm}
\end{proof}

%\subsection{$\cdimp$ and $\rho^-_{\mathsf{poly}}$}

We can alternatively use polynomial-time $s$-gales to define
polynomial-time dimension, $\cdimp$ (See \cite{Stull20}).
Nandakumar, Pulari, Akhil and Sarma \cite{OWFDimension} show that the
notions $\cdimp$ and $\Kpoly$ are distinct if one-way functions exist.

\begin{theorem} [Nandakumar, Pulari, Akhil and Sarma \cite{OWFDimension}]
  If one-way functions exist, then there are sequences 
  $X \in \Sigma^\infty$ such that $\Kpoly(X) < \cdimp(X)$.
\end{theorem}

We use this result along with Theorem \ref{thm:eqvlRhoKpoly} to show
that $\rho^-_{\mathsf{poly}}(X)$ and $\cdimp$ are distinct if One-way
functions exist.

\begin{corollary} 
  If one-way functions exist, then there are sequences $X \in
  \Sigma^\infty$ such that $\rho^-_{\mathsf{poly}}(X) < \cdimp(X)$.
\end{corollary}

\subsection{$\Kpoly$ and Polynomial-time Ku\v{c}era--G\'{a}cs}

We now extend the quasi-polynomial-time Ku\v{c}era--G\'{a}cs
theorem, Theorem \ref{thm:quasiPolynKG}, taking into account the $\Kpoly$ dimension of the sequence $X$. We
show a quasi-polynomial-time Ku\v{c}era--G\'{a}cs reduction that only
queries bits of polynomial-time random $R$ corresponding to the
$\Kpoly$ dimension of the sequence $X$.

\begin{restatable}{theorem}{KuceraGacsPolyExtended} \label{thm:KuceraGacsPolyExtended} 
	For all $X \in \Sigma^\infty$, there exists a polynomial-time random 
	$R \in \Sigma^\infty$ such that $X \leq_{t(n)} R$ 
	via $M$ with oracle use $u_n$ satisfying
	$
	 \liminf_{n \to \infty} ({u_n}/{n}) \leq \Kpoly(X) ,
	$
	
	where
	$t(n) : \N \to \N$
	is any time-constructible function such that $t(n) = \omega(n^k)$ for all $k \in \N$.
\end{restatable}
\begin{proof}
	Given $X \in \Sigma^\infty$, such that $\Kpoly(X) = s$. For any
	$i \in \N$, there exists an $s_i \in \Q$ such that $s_i > \Kpoly(X)$
	and $s_i - \Kpoly(X) < 2^{-i}$. Therefore, we have that there exists
	a $t_i(n) \in \poly$ such that for infinitely many $n_{i,j} \in \N$,
	$K_{t_i}(X \restr n_{i,j}) < s_i \cdot n_{i,j}$.
	
	First, we use a technique similar to the proof in Theorem
	\ref{thm:eqvlRhoKpoly} to combine these finite prefixes into a single
	string $Y \in \Sigma^\infty$ such that $X \leq_{t(n)} Y$ such that
	for infinitely many $n \in \N$, the oracle use
	$u_n <s \cdot n + o(n)$.
	
	For all $i \in \N$, define $m_i$ as follows $m_i = 0$ and for
	$i \geq 1$, $m_{i} = n_{i,j}$ for the least $j \in \N$ such that
	$m_{i}\geq (\sum_{k = 1}^{i - 1} m_k)^2$. Let $\pi_i$ be a prefix
	free program such that $U(\pi_i) = X \restriction m_i$ and
	$|\pi_i| \leq s_i\cdot m_i$.
	
	Let $\ell_i = \frac{1}{2} \log m_i$. Define $v_i$ =
	$0.u_i^1.0.u_i^2. \dots u_i^{l_i}.1$. Here
	$u_i^j = X[m_{i-1} + 2^{j - 1} : m_{i-1} + 2^j]$. Finally define
	$Y = v_1.\pi_1.v_2.\pi_2.v_3.\pi_3 \dots$.
	
	We can see that $X \leq_{T} Y$ via the machine $M$ given in
	Algorithm \ref{algm:1}. For any $n\in \N$, to produce $X \restr n$,
	$M$ uses time at most $t_i(n) + n + t_i(n^2)$, where $i$ is the
	least number such that $m_i > n$. As $t(n) \in \omega(\poly)$, we
	have that $X \leq_{t(n)} Y$.
	
	From the same analysis given in proof of Theorem \ref{thm:eqvlRhoKpoly}, we  have that for all $i \in \N$, the number of bits $M$ queries
	from $Y$ to produce $X \restr m_i$,
	$ \#(X \upharpoonright m_i, M^Y) \leq s \cdot m_i + 5 \cdot \sqrt{m_i} .$
	
	Now from the quasi-polynomial-time Ku\v{c}era--G\'{a}cs theorem
	(Theorem \ref{thm:quasiPolynKG}), we have that there exists an
	$R \in \Sigma^\infty$ such that $R$ is polynomial-time random and
	$Y \leq_{t(n)} R$.
	Additionally, to obtain the first $n$ bits of $Y$, we only need the first
	$n + o(n)$ bits of $R$.
	
	Now we combine both of these together. Given $X \in \Sigma^\infty$,
	we can first reduce $Y \leq_{t(n)} R$ and then $X \leq_{t(n)} Y$.
	The total time taken in the reduction is still
	$\mathcal{O}(t(n))$. Additionally, we
	have that for infinitely many $n \in N$, the oracle use from $Y$ is
	less than $s\cdot n + o(n)$, and therefore the oracle use from $R$,
	$u_n < s\cdot n + o(n)$. Therefore, we have that
	$\liminf_{n \to \infty} \frac{u_n}{n} \leq s$.
	
\end{proof}

\section{ Ku\v{c}era--G\'{a}cs for finite-state reductions}

In the previous sections, we showed that starting from
polynomial-time randoms and using
quasi-polynomial-time reductions, any sequence
\(X \in \Sigma^\infty\) can be extracted.
To further clarify the trade-off between reduction power
and source randomness, we prove a negative result showing
that the reduction cannot be too simple, even when we
start with a fairly weak randomness notion. 
The reducibility notion we use is that of finite-state
reductions, implemented via finite-state transducers.
The weak randomness notion
we consider is the class of normal sequences, which is a superset of
the class of polynomial-time randoms. 
A sequence \(N \in \Sigma^\infty\) is \emph{normal} if every
block of length \(\ell\) occurs with uniform limiting frequency
\(|\Sigma|^{-\ell}\).
Normality captures the notion of finite-state randomness \cite{BourkeHitchcockVinodchandran}, and it can be
characterized both by incompressibility via
information-lossless finite-state compressors
\cite{Dai2001}, and by the failure of finite-state
martingales \cite{SchSti72}.

Thus, the Ku\v{c}era--G\'{a}cs analogue for finite-state
reductions asks: \emph{Let $\Gamma$ and $\Sigma$ be finite alphabets.
	For every sequence $X \in \Gamma^\infty$, does there exist a normal
	sequence $N \in \Sigma^\infty$ from which $X$ can be computed by a
	finite-state transducer?}
We show that this is \emph{not} the case.

In particular, we prove that if a sequence $X$ is finite-state
reducible to a normal sequence $N$, then $X$ satisfies an
\emph{invariant} property: the empirical frequencies of symbols in $X$
must converge. Since it is easy to construct sequences without
limiting symbol frequencies (see Example~\ref{eg:one}),
such sequences are not finite-state reducible to a normal
sequence.

\begin{remark}
	Unlike earlier sections, we use two (possibly distinct) finite alphabets, $\Sigma$ for input and $\Gamma$ for output. This is a generalisation of the single-alphabet case ($\Sigma$ = $\Gamma$).
	This distinction is kept because finite-state measures of randomness
	such as normality \cite{Schmidt61} and finite-state dimension \cite{Pulari25} of the same real is sensitive to the base of representation used. 
	
\end{remark}
\subsection{Finite-state Reductions}

We define finite-state transducers and finite-state reductions.
\begin{definition} 
	A finite transducer  is a 6-tuple \( T = (Q, \Sigma,
        \Gamma, q_0, \delta ,\tau ) \),   
        where \( Q \) is a finite set of states, \( \Sigma \) is the finite input alphabet,
		\( \Gamma \) is the finite output alphabet,
		\( q_0 \in Q \) is the initial state,
		 \( \delta: Q \times \Sigma \to Q\) is the transition function, and
		 \( \tau: Q \times \Sigma \to \Gamma^*\) is the output function.
	\end{definition}

\begin{definition}
	The extended transition function 
	\( \hat{\delta}: Q \times \Sigma^* \to Q \) 
	is defined recursively by
	\(
	\hat{\delta}(q,\lambda)=q
	\quad\text{and}\quad
	\hat{\delta}(q,wa)=\delta(\hat{\delta}(q,w),a)
	\text{ for } w\in\Sigma^*,\, a\in\Sigma.
	\)

	The extended output function 
	\( \hat{\tau}: Q \times \Sigma^* \to \Gamma^* \) 
	is defined recursively by 
	\( \hat{\tau}(q,\lambda)=\epsilon \) and 
	\( \hat{\tau}(q,wa)=\hat{\tau}(q,w)\cdot
	\tau(\hat{\delta}(q,w),a) \) 
	for \( w\in\Sigma^*, a\in\Sigma \).
\end{definition}

\begin{definition} We say that a sequence $X \in \Gamma^\infty$ is
	finite-state reducible to $Y \in \Sigma^\infty$ (denoted
	$X\le_{FS} Y$) iff there exists a finite state transducer $T = $
	\( (Q, \Sigma, \Gamma, q_0, \delta ,\tau ) \) such that
	$ \lim_{\m \to \infty} \hat{\tau}(q_0 , Y \restriction m) = X.$
\end{definition}

\begin{definition}
  Given a $u \in \Sigma^\ell$, and $w \in \Sigma^{k.\ell}$,
  define
  $$P(u, w) = \frac{\Big|\{i < k : w[i\ell : ((i+1)\ell) - 1] = u\}
    \Big|}{k}.$$
\end{definition}

\begin{definition}
  Given a $Y \in \Sigma^\infty$, $Y$ is normal if for all
  $\ell \in \N$ and for all $u \in \Sigma^\ell$,
  $\lim_{k \to \infty} P(u, Y \restriction k \ell) =
  |\Sigma|^{-\ell}.$
\end{definition}

We now show an example of a binary sequence that does not have a
limiting distribution of ones (and hence of zeroes).

\begin{restatable}{example}{examplezeoone}\label{eg:one}
	Consider a sequence $X \in \{0,1\}^\infty$ defined as follows. $X[0]=0$, and for all $n
	\ge 1$, the bits $X[n!~:~(n+1)!-1]$ are all ones if $n$ is
	odd, and all zeroes if $n$ is even.
	
	Then, $\liminf_{n } P(1, X\restr n)  = 0$ and $\limsup_{n } P(1, X\restr n)  = 1$. So, $\lim_{n \to \infty} P(1, X\restr n)$
	does not exist.
\end{restatable}

\subsection{Impossibility of Ku\v{c}era--G\'{a}cs for finite-state reductions}

The following lemma shows that for any finite-state
transducer $T$ on a normal sequence, there exists a stationary
distribution $\pi : Q \to [0,1]$ such that the limiting frequency of each
transition $(q,a) \in Q \times \Sigma$ is $\pi(q)/|\Sigma|$.

\begin{lemma} [ Schnorr and Stimm \cite{SchSti72}] \label{lem:schnorrStimm} Let
  $T=(Q,\Sigma,\Gamma, q_0, \delta,\tau)$ be a finite-state
  transducer. Then there is a probability distribution
  $\pi : Q \to [0,1]$, such that for all normal sequences
  $Y \in \Sigma^\infty$, all $q\in Q$, and all
  $u\in\Sigma$,
  $$\frac{\pi(q)}{|\Sigma|}=\lim_{m\to\infty}\frac{1}{m}\Bigg|\Bigg|\Big\{k\le
  m: \hat{\delta}(q_0,Y[0:k]))=q\land Y[k+1]=u\Big\}\Bigg|\Bigg|$$
\end{lemma}

We show that if a sequence $X\le_{FS} Y$, and $Y$ is normal, then
the limiting distribution of the frequency of alphabets $a$ in $X$
must converge.

\begin{restatable}{lemma}{convergenceLemma}\label{lem:convergenceLemma}
 For all
  $ X\in \Gamma^{\infty},$ and $ Y\in \Sigma^\infty,$ if $Y$ is normal
  and $X\le_{FS} Y$, then there exists a probability distribution
  $p: \Gamma \to [0,1]$ such that
	 $\lim_{n\to\infty}P(a, X\upharpoonright n) = p(a).$
	
\end{restatable}

\begin{proof}
	Let $T=(Q,\Sigma,\Gamma,q_0,\delta,\tau)$ be a finite-state transducer such that
	$X \le_{FS} Y$, and assume that $Y$ is normal.  Let
	$
	\ell := \max_{q\in Q}\max_{u\in\Sigma} |\tau(q,u)|
	$
	be the maximum output length of a single transition in $T$.
	
	Fix $a\in\Gamma$. For each $n\in\mathbb N$, let $m=m(n)$ be the least integer such that
	$
	n \le \bigl|\widehat{\tau}(q_0, Y\upharpoonright m)\bigr|.
	$
	
	Then $X\upharpoonright n$ is a prefix of
	$\widehat{\tau}(q_0, Y\upharpoonright m)$, and
	$
	n \le 
	\bigl|\widehat{\tau}(q_0, Y\upharpoonright m)\bigr|
	< n + \ell.
	$
	%	Consequently, the difference between symbol frequencies computed
	%	using $X\upharpoonright n$ and those computed using
	%	$\widehat{\tau}(q_0, Y\upharpoonright m)$ is $O(1/n)$.
	
	Now, for each transition $(q,u)\in Q\times\Sigma$, let
	$
	N_m(q,u) := \#\bigl((q,u),\,Y\upharpoonright m\bigr)
	$
	denote the number of times the transition $(q,u)$ is taken while processing
	$Y\upharpoonright m$. Let  $\#_a\bigl(\tau(q,u)\bigr)$ denote the number of times $a$ occurs in in $\tau(q,u)$.

	Every symbol $a\in\Gamma$ appearing in $X\upharpoonright n$
	is produced by some transition $(q,u)\in Q\times\Sigma$ of $T$. So, the number of occurrences of $a$ in the output $X \restr n$, upto to an additive error of $\ell$ is
	\[
	\#(a, X\upharpoonright n)
	= 
	\sum_{(q,u)\in Q\times\Sigma}
	N_m(q,u)\cdot \#_a\bigl(\tau(q,u)\bigr)
	\;.
	\]
	and hence upto an additive error of $\ell/n = o(1)$, we have
	\begin{align*}
		P(a, X\upharpoonright n)
		&=
		\frac{\#(a, X\upharpoonright n)}{n} 
		= \frac{m}{n} 
		\sum_{(q,u)\in Q\times\Sigma}
		\#_a\bigl(\tau(q,u)\bigr)\cdot\,
		\frac{N_m(q,u)}{m}.
	\end{align*}
	Also,
	\[
	\bigl|\widehat{\tau}(q_0, Y\upharpoonright m)\bigr|
	=
	\sum_{(q,u)\in Q\times\Sigma}
	N_m(q,u)\cdot|\tau(q,u)|,
	\]
	so upto an additive error of $\ell/m = o(1)$,
	\[
	\frac{n}{m}
	=
	\sum_{(q,u)\in Q\times\Sigma}
	\frac{N_m(q,u)}{m}\cdot|\tau(q,u)|
	\;.
	\]
	
	Combining, and ignoring $o(1)$ error terms,
	\begin{align*}
		P(a, X\upharpoonright n)
		&=
		\frac{
			\sum\limits_{(q,u)\in Q\times\Sigma}
			\#_a\bigl(\tau(q,u)\bigr)\cdot\,
			\frac{N_m(q,u)}{m}} {\sum\limits_{(q,u)\in Q\times\Sigma}
			\frac{N_m(q,u)}{m}\cdot|\tau(q,u)|}\;.
	\end{align*}

	By the Schnorr--Stimm lemma (Lemma \ref{lem:schnorrStimm}),
	\[
	\lim_{m\to\infty}\frac{N_m(q,u)}{m}=\frac{\pi(q)}{|\Sigma|}
	\qquad\text{for every }(q,u)\in Q\times\Sigma.
	\]
	Therefore the limit
	$
	\lim_{n\to\infty} P(a,X\upharpoonright n)
	$
	exists and equals
	\begin{align*}
		p(a)
		&=
		\frac{
			\sum\limits_{(q,u)\in Q\times\Sigma}
			\#_a\bigl(\tau(q,u)\bigr)\cdot\pi(q)
		}{
			\sum\limits_{(q,u)\in Q\times\Sigma}
			|\tau(q,u)|\cdot\pi(q)
		}.
	\end{align*}
	
	Finally,
	\[
	\sum_{a\in\Gamma}
	\sum_{(q,u)\in Q\times\Sigma}
	\#_a\bigl(\tau(q,u)\bigr)\,\pi(q)
	=
	\sum_{(q,u)\in Q\times\Sigma}
	|\tau(q,u)|\,\pi(q),
	\]
	so $\sum_{a\in\Gamma} p(a)=1$. Hence $p$ is a probability distribution on
	$\Gamma$.
\end{proof}

We show that an analogue of the Ku\v{c}era--G\'{a}cs theorem does not hold in
the finite-state setting.

\begin{theorem}
  There exists an $X \in \Gamma^\infty$, such that for any normal
  $N \in \Sigma^\infty$, $X \not\leq_{FS} N$.
\end{theorem}
\begin{proof}
  Take any $X \in \Gamma^\infty$ such that for some $a \in \Gamma$,
  $\lim_{n \to \infty} P(a, X \restriction n)$ does not converge.
	
  If for some normal $N\in \Sigma^\infty$, $X \leq_{FS} N$, from Lemma
  \ref{lem:convergenceLemma}, we have that
  $\lim_{n \to \infty} P(a, X \restriction n)$ exists, which is a
  contradiction.
\end{proof}

\section{Acknowledgements}
We thank Laurent Bienvenu, Elvira Mayordomo, Joseph S. Miller, Subin Pulari, and Alexander Shen for helpful discussions and comments.
We also thank the anonymous reviewers for their constructive comments, pointing out errors, providing missing references and suggesting open questions.

\section{Open Problems}

\begin{itemize}
	\item Can we improve the quasi-polynomial-time bound in Theorem \ref{thm:quasiPolynKG} to polynomial-time. Does any inherent
	complexity-theoretic barriers prevent such a strengthening ? 
	\item Does an analogue of Theorem \ref{thm:eqvlRhoKpoly} hold for strong dimension. Can we show that $\rho^+_{\mathsf{poly}}(X) = \mathcal{K}^{str}_{\mathsf{poly}}(X)?$ (See \cite{Doty08} and \cite{dercjournal} for definitions).
\end{itemize}

\bibliography{main}

\end{document}